\documentclass[12pt,reqno]{amsart}

\usepackage{amsmath}
\usepackage{latexsym}
\usepackage{amsfonts}
\usepackage{amssymb}
\usepackage{color}
\usepackage{epsfig}
\usepackage{bbm,dsfont}
\usepackage{textcomp}

\newtheorem{proposition}{Proposition}
\newtheorem{theorem}{Theorem}

\theoremstyle{definition}

\newtheorem{example}{Example}
\newtheorem{definition}{Definition}

\newcommand{\R}{\mathbb R} %real
\newcommand{\C}{\mathbb C} %complex
\newcommand{\half}{\tfrac{1}{2}} %half
\renewcommand{\mod}[1]{|#1|} %modulus

\newcommand{\hi}{\mathcal{H}} %Hilbert space
\newcommand{\eh}{\mathcal{E(H)}} %effects
\newcommand{\ip}[2]{\left\langle\,#1\,|\,#2\,\right\rangle} %inner product
\newcommand{\no}[1]{\left\|#1\right\|} %norm
\newcommand{\id}{I} %identity operator
\newcommand{\nul}{O} %null operator

\newcommand{\va}{\mathbf{a}} %a
\newcommand{\vb}{\mathbf{b}} %b
\newcommand{\vg}{\mathbf{g}} %g
\newcommand{\vn}{\mathbf{n}} %n
\newcommand{\vm}{\mathbf{m}} %m
\newcommand{\vsigma}{\boldsymbol{\sigma}} %sigma
\newcommand{\vnull}{\mathbf{0}}%null vector

\newcommand{\coex}{\;\text{\rm \textmarried}\;}

\newcommand{\salg}{\mathcal{F}} %sigma-algebra

\newcommand{\A}{A}%generic observable
\newcommand{\B}{B}%generic observable
\newcommand{\F}{\mathsf{F}}%generic observable
\newcommand{\G}{\mathsf{G}}%generic joint observable
\renewcommand{\S}{\mathfrak{S}}

\newcommand{\bx}{b_{\parallel}} %component of b parallel to a
\newcommand{\by}{b_{\perp}} %component of b perpendicular to a
\newcommand{\allow}{\mathfrak{A}} %allowed region

\begin{document}

\title{Coexistence of qubit effects}

\author[Stano]{Peter Stano}
\address{Peter Stano, Research Center for Quantum Information, Slovak Academy of Sciences, D\'ubravsk\'a cesta 9, 845 11 Bratislava, Slovakia}
\email{peter.stano@savba.sk}

\author[Reitzner]{Daniel Reitzner}
\address{Daniel Reitzner, Research Center for Quantum Information, Slovak Academy of Sciences, D\'ubravsk\'a cesta 9, 845 11 Bratislava, Slovakia}
\email{daniel.reitzner@savba.sk}

\author[Heinosaari]{Teiko Heinosaari}
\address{Teiko Heinosaari, Research Center for Quantum Information, Slovak Academy of Sciences, D\'ubravsk\'a cesta 9, 845 11 Bratislava, Slovakia and Department of Physics, University of Turku, 20014 Turku, Finland}
\email{heinosaari@gmail.com}

\begin{abstract}
We characterize all coexistent pairs of qubit effects. This gives an exhaustive description of all pairs of events allowed, in principle, to occur in a single qubit measurement. The characterization consists of three disjoint conditions which are easy to check for a given pair of effects. Known special cases are shown to follow from our general characterization theorem.
\end{abstract}

\maketitle

\section{Introduction}

Optimal solutions for quantum information processing tasks typically require observables that cannot be described by single selfadjoint operators but are formalized as positive-operator-valued measures (POVMs). Generally, an element of an observable, called an \emph{effect}, can be any positive operator bounded by the identity operator. For instance, an optimal observable for unambiguous discrimination of two non-orthogonal pure states has three elements and none of them is a projection \cite{Chefles00}. Another example is provided by informationally complete observables, which do not have any non-trivial projections as their elements \cite{BuCaLa95}. 

It is well known that two projections can be elements of a single observable if and only if they commute. This condition for effects to be parts of a single observable is called \emph{coexistence} \cite{FQMI83,SEO83}. Coexistence can be therefore viewed as a kind of natural generalization of commutativity. It is remarkable that two effects can be coexistent even if they do not commute, but a general criterion of coexistence is not known. This problem of characterizing coexistent effects, called the \emph{coexistence problem}, is the topic of this paper.

The coexistence of effects is connected to the theoretical limitations built inside the quantum theory, and the concept of coexistence provides a unifying framework for these kinds of issues. Indeed, many theoretical limitations, related both to the foundations and to quantum information processing tasks, can be seen as a consequence of (non-)coexistence of the relevant effects. For instance, the security of Bennett-Brassard 1984 (BB84) protocol \cite{BeBr84} relies on the non-coexistence of the corresponding effects. Moreover, assuming that the Bell inequality is violated, the coexistence of certain effects would lead to the possibility of superluminal communication \cite{Werner01}.

Coexistence (contrary to commutativity) also explains the possibility of unsharp joint measurements of complementary pairs of physical quantities, such as orthogonal spin components or path and interference of an atomic beam. A joint measurement of such pairs is possible only if an increased unsharpness is accepted, and a relevant coexistence condition can be then interpreted as a trade-off relation between the imprecisions of the corresponding measurements. Some recent investigations on this issue are reported, for instance, in \cite{AnBaAs05,BuSh06,BuHe07,LiLiYuCh07}.

In this work we give a complete characterization of the previously stated coexistence problem in the case of two qubit effects. In Sec.~\ref{sec:problem} we recall the coexistence problem in a precise formulation. In Sec.~\ref{sec:fundamental} we present the main result of this paper --- a characterization theorem of coexistent pairs of qubit effects. We also show that the already known special cases are easily recovered from our theorem. A detailed proof of the characterization theorem is given in the appendixes. In Appendix 1 we recall some general facts on coexistence which are needed in our investigation. Appendix 2 then concentrates on the details of the proof.

\section{Coexistence problem}\label{sec:problem}

Let $\hi$ be a complex separable Hilbert space.  An operator $A$ on $\hi$ is an \emph{effect} if 
\begin{equation*}
0\leq \ip{\psi}{A\psi}\leq 1
\end{equation*}
for all $\psi\in\hi$. In terms of operator inequalities this reads
\begin{equation*}
\nul\leq A\leq \id \, ,
\end{equation*}
where $\nul$ and $\id$ are the zero operator and the identity operator, respectively.
We denote by $\eh$ the set of effects.

An \emph{observable} $\G$ is a normalized-effect-valued measure, also called a positive-operator-valued measure (POVM). It is defined on a measurable space $(\Omega,\salg)$, where $\Omega$ is the set of measurement outcomes and $\salg\subseteq 2^\Omega$ is the $\sigma$-algebra of possible events. For each event $X$, the observable $\G$ attaches an effect $\G(X)$. If the system is in a vector state $\psi\in\hi$ and a measurement of $\G$ is performed, the probability of getting a measurement outcome $\omega$ belonging to an event $X$ is $\ip{\psi}{\G(X)\psi}$. Detailed explanations and many examples of this generalized description of quantum observables can be found in \cite{QTOS76, PSAQT82, OQP97, FQMEA02}. 

For a singleton set $\{\omega \}\subset\Omega$, we denote $\G_\omega\equiv\G(\{\omega\})$. If the set of measurement outcomes $\Omega$ is countable, then $\G$ is determined by the set of effects $\G_\omega$, $\omega\in\Omega$. Namely, a general effect $\G(X)$ corresponding to an event $X$ is recovered by formula
\begin{equation*}
\G(X)=\sum_{\omega\in X} \G_\omega \, .
\end{equation*}
In particular, an observable $\G$ with a finite number of measurement outcomes (say, $n$) can be described as a list $(\G_{\omega_1},\ldots,\G_{\omega_n})$. The POVM normalization condition then reads
\begin{equation*}
\sum_{i=1}^n \G_{\omega_i} = \id \, .
\end{equation*}
  
We can also look on the structure of observables from the other side: given a collection of effects, we can ask whether they originate in a single observable. This concept, called coexistence, was first studied by Ludwig \cite{FQMI83}.

\begin{definition}\label{def:definition of coexistence}
Effects $A,B,C,\ldots \in\eh$ are \emph{coexistent} if there exists an observable $\G:\salg\to\eh$ and events $X,Y,Z,\ldots \in\salg$ such that 
\begin{equation}\label{eq:definition of coexistence}
A=\G(X),\quad B=\G(Y), \quad C=\G(Z), \quad \ldots 
\end{equation}
If two effects $A$ and $B$ are coexistent, we denote $A\coex B$.
\end{definition}

It is essential to note that in Definition \ref{def:definition of coexistence} the events $X,Y,Z,\ldots$ need not be disjoint. 

As an example, let $\F$ be the symmetric informationally complete qubit observable consisting of four effects
\begin{eqnarray*}
\F_1 &=& \frac{1}{4} \left[ \id+\frac{1}{\sqrt{3}}(\sigma_x+\sigma_y+\sigma_z) \right] \, , \\
\F_2 &=& \frac{1}{4} \left[ \id+\frac{1}{\sqrt{3}}(-\sigma_x-\sigma_y+\sigma_z) \right] \, ,\\
\F_3 &=& \frac{1}{4} \left[ \id+\frac{1}{\sqrt{3}}(-\sigma_x+\sigma_y-\sigma_z) \right] \, , \\
\F_4 &=& \frac{1}{4} \left[ \id+\frac{1}{\sqrt{3}}(\sigma_x-\sigma_y-\sigma_z) \right] \, .
\end{eqnarray*}
The fact that $\F$ is an observable implies that the effects $\half (\id + \frac{1}{\sqrt{3}} \sigma_j )$, $j=x,y,z$ are coexistent. Indeed, we get
\begin{equation*}
\F(\{1,4\})=\F_1+\F_4=\half (\id + \frac{1}{\sqrt{3}} \sigma_x )
\end{equation*}
and similarly for the other two effects. Actually, this reasoning leads also to a proof of the informational completeness of $\F$ as we can conclude that a measurement of $\F$ provides the same information as three separate measurements of the orthogonal spin components. This example should be compared with the fact that any two projections $\half (\id +  \vn\cdot\vsigma )$ and $\half (\id +  \vm\cdot\vsigma )$ with $\vn\neq\pm\vm$ do not commute and hence are not  coexistent.

In this paper we concentrate on the following \emph{coexistence problem}:

\begin{quote}
Given an effect $A$, characterize all effects $B$ which are coexistent with it. 
\end{quote}

The following simple observation shows that when we are studying the coexistence of two effects (as in this paper), we can restrict ourselves to four outcome observables. 

\begin{proposition}\label{prop:G4}
Effects $A$ and $B$ are coexistent if and only if there exists an observable $\G$ with four outcomes $\{1,2,3,4\}$ such that 
\begin{equation}\label{eq:G4}
A=\G_1+\G_2,\quad B=\G_1+\G_3 \, .
\end{equation}
\end{proposition}

\begin{proof}
By definition, if a four outcome observable $\G$ satisfying Eq.~(\ref{eq:G4}) exists, then $A$ and $B$ are coexistent. Assume then that $A$ and $B$ are coexistent and let $\G:\salg\to\eh$ be an observable such that $A=\G(X), B=\G(Y)$ for some $X,Y\in\salg$. We denote $X'=\Omega\setminus X$ and $Y'=\Omega\setminus Y$, and we set $\widetilde{\G}_1=\G(X\cap Y)$, $\widetilde{\G}_2=\G(X\cap Y')$, $\widetilde{\G}_3=\G(X'\cap Y)$, and $\widetilde{\G}_4=\G(X'\cap Y')$. This defines an observable $\widetilde{\G}$ with the required properties.
\end{proof}

If $A$ is a projection (i.e. $A=A^2$), then the answer to the coexistence problem is simple and well known: an effect $B$ is coexistent with $A$ exactly when $AB=BA$. Generally, however, a characterization of coexistent effects is not known. In the next section we present a full solution to the coexistence problem in the case of a qubit system, i.e. two dimensional Hilbert space $\hi=\C^2$.
 
\section{Qubit effects and their coexistence}\label{sec:fundamental}

Qubit effects $A$ and $B$ can be parametrized by vectors $(\alpha,\va),(\beta,\vb)\in\R^4$ in the following way:
\begin{subequations}
\begin{eqnarray}
A&=&\frac{1}{2}(\alpha I+\va\cdot\vsigma),\qquad a \leq \alpha \leq 2-a\, ,
\label{definition of A}\\
B&=&\frac{1}{2}(\beta I+\vb\cdot\vsigma),\qquad b\leq \beta\leq 2-b \, .
\label{definition of B}
\end{eqnarray}
\label{definition of A and B}
\end{subequations}
Here $\boldsymbol{\sigma}\equiv(\sigma_1,\sigma_2,\sigma_3)$ is the vector of Pauli matrices, and we have denoted $a\equiv\no{\va}$, $b\equiv\no{\vb}$. Note that from Eqs.~(\ref{definition of A}) and (\ref{definition of B}) it follows that $a,b\leq 1$. 

We are now considering $A$ to be fixed and we are looking for all effects $B$ (hence all parameters $\beta$ and $\vb$), which are coexistent with $A$. In order to formulate the characterization theorem, we first introduce the following function $\S$ from $\eh$ to $[0,1]$,
\begin{equation}
\S(A)\equiv\S(\alpha,a):=\frac{1}{2}\left(a^2+\alpha(2-\alpha)-\sqrt{(\alpha^2-a^2)[(2-\alpha)^2-a^2]} \right) \, .
\label{eq:sharpness}
\end{equation}
The following properties of $\S$ are easy to confirm:
\begin{itemize}
\item[(a)] $\S$ is continuous;
\item[(b)] $\S(\id-A)=\S(A)$;
\item[(c)] $\S(UAU^\ast)=\S(A)$ for every unitary operator $U$;
\item[(d)] $\S(A)=1$ if and only if $A$ is a non-trivial projection (i.e. $A^2=A$ and $\nul\neq A\neq \id$);
\item[(e)] $\S(A)=0$ if and only if $A$ is a trivial effect (i.e. $A=\lambda\id$ for some $0\leq\lambda\leq 1$).
\end{itemize}
Due to these properties, we interpret the number $\S(A)$ as a quantification of the \emph{sharpness} of $A$. Naturally, $1-\S(A)$ is then related to the \emph{unsharpness} of $A$. 

For simplicity, we formulate the main theorem below in the case of $0<\alpha\leq 1$ and $0<\beta\leq 1$. We note that if $A$ is defined by parameters $\alpha$ and $\va$, then $\id-A$ corresponds to $2-\alpha$ and $-\va$. As shown in Proposition \ref{prop:equivalent} in Appendix 1, the coexistence of $A$ and $B$ is equivalent to the coexistence of $\id-A$ and $B$. Therefore, the cases when $\alpha>1$ or $\beta>1$ can be recovered easily from the main theorem.  

It is useful to note that for effects $A$ satisfying $0\leq\alpha\leq 1$, we have
\begin{equation}
\S(a,\alpha)\leq \S(\alpha,\alpha)=\alpha \, ,
\end{equation}
and the inequality is strict whenever $a\neq\alpha$. Therefore, the value of the parameter $\alpha$ gives an upper bound for the sharpness of $A$. 

It follows from Proposition \ref{prop:unitary} in Appendix 1 that only the relative angle between $\va$ and $\vb$ is relevant for the coexistence of $A$ and $B$ --- not their absolute directions. In the following it is thus convenient to denote by $\bx$ the component of $\vb$ in the direction of $\va$, and $\by$ the length of the projection of $\vb$ in the plane perpendicular to $\va$. For a given $A$, the coexistence of $B$ with $A$ then depends on parameters $\bx$, $\by$, and $\beta$.

\begin{theorem}\label{th:fundamental}
An effect $B$ is coexistent with $A$ if and only if it falls into one of the following three disjoint cases: 
\begin{itemize}
\item[(C1)] if $\beta\leq 1-\S(A)$, then $A\coex B$ irrespectively of $\vb$;
\item[(C2)] if $\beta>1-\S(A)$ and $|\bx-b_0|\geq w$, then $A\coex B$;
\item[(C3)] if $\beta>1-\S(A)$ and $|\bx-b_0|< w$, then $A\coex B$ if and only if
 \begin{equation}
\by\leq \by^{max}.
\end{equation}
\end{itemize}
Here we have denoted
\begin{eqnarray}
\by^{max} &=& \frac{1}{2a}\sqrt{[(2-\alpha)^2-a^2]\{a^2-[a(\bx-b_0)+(1-\beta)]^2\}}\nonumber\\
 &&+\frac{1}{2a}\sqrt{[\alpha^2-a^2]\{a^2-[a(\bx-b_0)-(1-\beta)]^2\}}\, ,
 \label{nontrivial condition}\\
b_0&=&\frac{1}{a}(1-\alpha)(1-\beta)\, ,\\
w&=&\frac{1}{a}\sqrt{(1-\alpha)^2-\beta[(1-\alpha)^2+1-a^2]+\beta^2}\, .
\end{eqnarray}
\end{theorem}

Among the three different situations (C1)--(C3), only in the last one does the coexistence of $A$ and $B$ impose a nontrivial\footnote{By trivial restrictions we mean the inequalities in \eqref{definition of B}. As we have assumed that $\beta\leq 1$, the trivial restrictions are equivalent to the condition $b\leq\beta$.} restriction on the length of vector $\vb$. Namely, the condition $|\bx-b_0|< w$ implies that
\begin{equation}
b^2=\by^2 + \bx^2 \leq (\by^{max})^2+\bx^2<\beta^2.
\label{inside circle}
\end{equation}
The last inequality in Eq.~\eqref{inside circle} is proved at the end of Appendix 2, where we also show that the direction of $\vb$ in which the length of $\vb$ is most restricted is determined by the condition $\bx=b_0$. 
%Since for fixed $\beta$ the shorter $\vb$, the smaller the sharpness $\S(B)$, we interpret the directions along these shortest vectors to represent effects which are mostly restricted if we require coexistence with $A$.

The division of the coexistence condition to the three disjoint cases (C1)--(C3) can be intuitively understood in the following way. The first class (C1) consists of those effects $B$ for which $\beta$ [and consequently, the sharpness $\mathfrak{S}(B)$] is so small that with any choice of $\vb$, the coexistence of $A$ and $B$ is attained. If $\beta$ is above the given threshold $1-\S(A)$ (the unsharpness of $A$), then for some angles between $\va$ and $\vb$, the length of $\vb$ is restricted if $B$ is to coexist with $A$. Namely, there exists an interval for $\bx$, in which cases the length of $\vb$ is limited. The center of the interval is $b_0$, which represents the most strict restriction, and the width of the interval is $2w$. The second class (C2) then consists of those effects $B$ for which $\bx$ is outside the interval and which are coexistent with $A$ even if their sharpness would be the highest possible [i.e.~$\S(B)=\beta$]. The third class (C3) represents effects for which their sharpness is nontrivially restricted if they are to coexist with $A$. 

In Fig.~\ref{fig:examples} we present four illustrative examples. Fig.~\ref{fig:examples}(a) demonstrates the case (C1) where $\beta<1-\mathfrak{S}(A)$, and hence all effects with this $\beta$ coexist with $A$. In Fig.~\ref{fig:examples}(b) we keep the parameters $\alpha$ and $a$ unchanged while $\beta$ is enlarged such that $\beta>1-\S(A)$. The interval with nontrivial restriction on the length of vectors $\vb$ appears -- for $\bx$ outside this interval (C2) applies, while for $\bx$ inside, (C3) applies. Note that the center of the interval is not zero. In Fig.~\ref{fig:examples}(c) we have $\beta=1$ and now the interval is centered at zero, meaning that the restriction on the sharpness of $B$ is most strict if $\va$ and $\vb$ are orthogonal. Furthermore, $w=1$ and thus (C3) covers all the possible cases. In Fig.~\ref{fig:examples}(d) we have $a=\alpha$, which means that $A$ is a multiple of a projection. Nonzero $b_0$ results in a clearly visible asymmetry in the picture. 

\begin{figure}
\psfig{scale=0.8, file=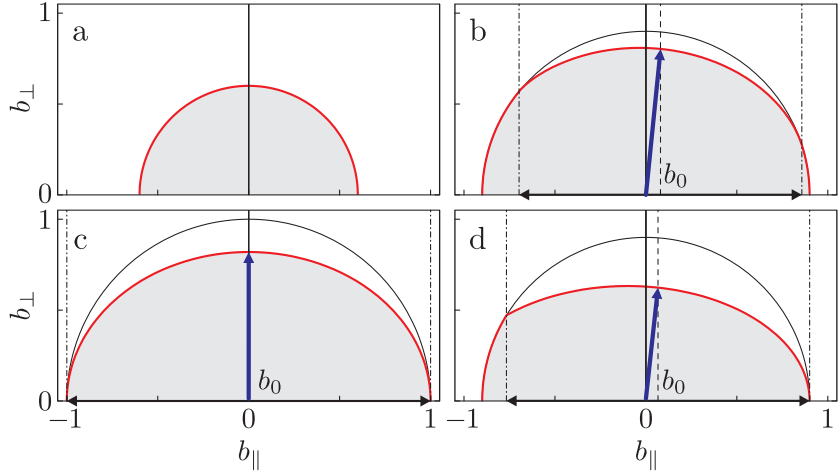}
\caption{Examples of the specification of effects $B$ which coexist with a given effect $A$. On each figure, parameters $\alpha$, $a$, and $\beta$ are fixed, while the $\vb$ vector components $\bx$ and $\by$ are on $x$ and $y$ axis, respectively. The thick red line denotes the boundary --- $A$ and $B$ coexist if and only if the vector $\vb$ is inside the shaded region. The thin black circle represents the condition on $B$ to be an effect, $b\leq\beta$. The blue vector represents the shortest vector $\vb$ lying on the boundary --- its projection to the $x$ axis equals $b_0$. The interval $[b_0-w,b_0+w]$, where a nontrivial restriction on the length of the allowed vectors $\vb$ exists, is denoted by the vertical dot-dashed lines and a black arrow on the $x$ axis. The following parameters are used in the pictures: $\alpha=0.6$ in every figure and (a) $a=0.5$, $\beta=0.6$; (b) $a=0.5$, $\beta=0.9$; (c) $a=0.5$, $\beta=1$; and (d) $a=0.6$, $\beta=0.9$.
\label{fig:examples}}
\end{figure}

In the following examples we demonstrate that the known special cases of coexistence conditions follow easily from Theorem \ref{th:fundamental}.

\begin{example}\label{ex:busch}
Assume that $\alpha=\beta=1$. Using property (e) of $\S(A)$, we see that the coexistence condition (C1) holds if and only if $\va=\vnull$. Whenever $\va\neq\vnull$, we have to look at the extra conditions in (C2) and (C3). We get $b_0=0$, $w=1$, and therefore (C2) occurs only whenever $|\bx|=1$. On the other hand, (C3) is satisfied when $|\bx|<1$ and
\begin{equation}\label{eq:busch1}
\by^2\leq (1-a^2)(1-\bx^2)\, .
\end{equation}
Putting $a=0$ in Eq.~\eqref{eq:busch1} we see that this inequality describes the correct solution also for the case $\va=\vnull$. The case $\bx=1$ is also recovered from Eq.~\eqref{eq:busch1} as $\bx=1$ implies that $\by=0$. In conclusion,  the inequality \eqref{eq:busch1} covers all the cases and we can write it in the symmetric form
\begin{equation}\label{eq:busch3}
a^2+b^2\leq 1 + (\va\cdot\vb)^2\, .
\end{equation}
This result has been first derived by Busch \cite{Busch86} in an equivalent form 
\begin{equation}\label{eq:busch2}
\no{\va+\vb}+\no{\va-\vb}\leq 2\, .
\end{equation}
Other derivations of this condition have been recently given in \cite{AnBaAs05} and \cite{BuHe07}.
\end{example}

\begin{example}\label{ex:chinese}
Assume that $\beta=1$ and $\va\perp\vb$. The first coexistence condition (C1) holds if and only if $\va=\vnull$. Since $\va\perp\vb$, we have $\bx=0$ and $\by=b$, while $b_0=0$ and $w=1$ due to $\beta=1$. Therefore (C2) does not occur and $\va\neq\vnull$ leads to the case (C3) which now reads
\begin{equation}\label{eq:chinese}
b\leq\frac{1}{2}\sqrt{(2-\alpha)^2-a^2}+\frac{1}{2}\sqrt{\alpha^2-a^2} \, .
\end{equation}
Putting $a=0$ in Eq.~\eqref{eq:chinese} we notice that this condition also covers the case $\va=\vnull$. This result has been derived by Liu et al. \cite{LiLiYuCh07}. 
\end{example}

\begin{example}\label{ex:molnar}
Assume that $a=\alpha$ and $b=\beta$, which means that the effects $A$ and $B$ are scalar multiples of projections. We now have $b_0=(1-\alpha)(1-\beta)/\alpha$ and $w=\mod{\alpha+\beta-1}/\alpha$. Since $b=\beta$, according to Eq.~\eqref{inside circle}, the effects $A$ and $B$ coexist if either (C1) or (C2) is satisfied. The condition (C1) holds if and only if $\alpha+\beta\leq 1$,
while in the case $\alpha+\beta>1$ the inequality in (C2) holds if either $\bx\geq\beta$ or
\begin{equation}\label{eq:molnar}
\va\cdot\vb\leq 2-2\alpha-2\beta+\alpha \beta\, . 
\end{equation}
The first case $\bx\geq\beta$ means that $\bx=b$ and thus $\vb$ is parallel to $\va$. In the second case, we notice that $\va\cdot\vb\leq ab\leq\alpha\beta$ and hence $\alpha+\beta\leq 1$ implies Eq.~(\ref{eq:molnar}). Therefore, the inequality (\ref{eq:molnar}) characterizes all the coexistent effects $A$ and $B$ having nonparallel vectors $\va$ and $\vb$. This inequality is also easily obtained from Lemma 2 of Moln\'ar \cite{Molnar01b}.
\end{example}

\section{Conclusion}

We have studied the coexistence problem of two qubit effects, i.e., the question of when two effects can be parts of a single observable. We have solved the problem by providing simple criteria (C1)--(C3), which, taken together, are necessary and sufficient to guarantee the coexistence. We have shown that the known special cases follow straightforwardly from our general coexistence theorem. 

We expect that many theoretical limitations, related both to the foundations and to quantum information processing tasks, can be seen as resulting from (non-)coexistence of effects. The concept of coexistence provides a natural unifying framework for these kinds of questions. The consequences of our main result, Theorem \ref{th:fundamental}, are yet to be found out. 

Finally, we remark that a paper by Busch and Schmidt \cite{BuSc08} was published  simultaneously on the arXiv with an earlier version of this paper. These authors solve the same problem independently with a different method. The final results have yet to be compared.

Recently, a third solution was published by Yu \emph{et al.}~\cite{YuLiLiOh08}. The connection between all these three approaches will be elaborated in a later work.

\section*{Appendix 1: General observations on coexistence}

In this section we list some simple general observations which are needed in the proof of Theorem \ref{th:fundamental}.

\begin{proposition}\label{prop:equivalent}
Let $A,B\in\eh$. The following conditions are equivalent:
\begin{itemize}
\item[(i)] $A$ and $B$ are coexistent;
\item[(ii)] $A$ and $\id-B$ are coexistent;
\item[(iii)] $\id-A$ and $B$ are coexistent;
\item[(iv)] $\id-A$ and $\id-B$ are coexistent.
\end{itemize}
\end{proposition}

\begin{proof}
It is enough to prove that  (i) implies (ii). The other implications follow by applying this to different combinations of $A$ and $\id-A$ with $B$ and $\id-B$.

Assume that $A$ and $B$ are coexistent and that $\G$ is a four outcome observable satisfying Eq.~(\ref{eq:G4}). We define another four outcome observable $\widetilde{\G}$ by 
\begin{equation*}
\widetilde{\G}_1:=\G_2,\quad \widetilde{\G}_2:=\G_1, \quad \widetilde{\G}_3:=\G_4,\quad \widetilde{\G}_4:=\G_3 \, .
\end{equation*}
Then 
\begin{equation*}
\widetilde{\G}_1+\widetilde{\G}_2=\G_2+\G_1=A
\end{equation*}
and
\begin{equation*}
\widetilde{\G}_1+\widetilde{\G}_3=\G_2+\G_4=\G_2+\id-\G_1-\G_2-\G_3=\id-B \, .
\end{equation*}
Thus, $A$ and $\id-B$ are coexistent.
\end{proof}

\begin{proposition}\label{prop:unitary}
Let $A,B\in\eh$ and $U$ a be unitary operator on $\hi$. The following conditions are equivalent:
\begin{itemize}
\item[(i)] $A$ and $B$ are coexistent;
\item[(ii)] $UAU^\ast$ and $UBU^\ast$ are coexistent.
\end{itemize}
\end{proposition}

\begin{proof}
It is enough to prove that  (i) implies (ii) as the other implication is obtained from this by applying $U^\ast$ instead of $U$. Assume that $A$ and $B$ are coexistent and that $\G$ is  a four outcome observable satisfying Eq.~(\ref{eq:G4}). Then $U\G_j U^\ast$ is a four outcome observable which satisfies a similar relation for observables $UAU^\ast$ and $UBU^\ast$.
\end{proof}

\begin{proposition}\label{prop:convex}
Let $A,B,C$ be effects such that $A$ is coexistent with both $B$ and $C$. Then for any $0\leq\lambda\leq 1$, the effects $A$ and $\lambda B + (1-\lambda) C$ are coexistent.
\end{proposition}

\begin{proof}
Let $\G^1$ be a four outcome observables satisfying the condition (\ref{eq:G4}) for $A$ and $B$, and $\G^2$ another four outcome observable satisfying a similar condition for $A$ and $C$. Let $\G$ be an observable defined as $\G_j = \lambda \G^1_j + (1-\lambda) \G^2_j$ for $j=1,2,3,4$. Then
\begin{eqnarray*}
\G_1+\G_2 &=& \lambda (\G^1_1 +  \G^1_2 ) + (1-\lambda)  (\G^2_1 + \G^2_2)\\
&=& \lambda A + (1-\lambda)  A = A
\end{eqnarray*}
and
\begin{eqnarray*}
\G_1+\G_3 &=& \lambda (\G^1_1 +  \G^1_3 ) + (1-\lambda)  (\G^2_1 + \G^2_3)\\
&=& \lambda B + (1-\lambda)  C \, .
\end{eqnarray*}
Hence, the effects $A$ and $\lambda B + (1-\lambda) C$ are coexistent.
\end{proof}

\begin{proposition}\label{prop:scaling}
Let $A$ and $B$ be coexistent effects. Then $A$ is coexistent with $\lambda B$ for any $0\leq\lambda\leq 1$.
\end{proposition}

\begin{proof}
Choose $C=\nul$ in Proposition \ref{prop:convex}.
\end{proof}

\section*{Appendix 2: Proof of the characterization theorem}

In this section we give a detailed proof of Theorem \ref{th:fundamental}. We first formulate the question whether two effects given in Eqs.~\eqref{definition of A} and \eqref{definition of B} coexist as a geometric problem. We then characterize its solution for particular boundary (i.e. limiting) cases. In the last step we identify and analyze each possible way of how a boundary case can occur. We find that a boundary case can happen in only two ways: the first way leads to (C1) and (C2) and the second leads to the (C3) condition.

\subsection*{Step 1: Formulation of the coexistence condition as an intersection requirement for four circles}

We first shortly recall the formulation of the coexistence condition as an intersection requirement for four circles \cite{Busch86,BuHe07}. As shown in Proposition \ref{prop:G4}, the coexistence of $A$ and $B$ is equivalent to the existence of a four outcome observable $\G$. This, in turn, is equivalent to the existence of a single effect $\G_1$ satisfying the following operator inequalities \cite{SEO83}:
\begin{equation}
\G_{1}\geq \nul,\quad \G_{1}\leq A, \quad \G_{1} \leq B, \quad \id+\G_{1} \geq A + B\, .
\label{conditions on G++}
\end{equation}

We parametrize $\G_1$ in the same way as $A$ and $B$ in Eqs.~\eqref{definition of A} and \eqref{definition of B},
\begin{equation}\label{eq:definition of G++}
\G_{1}=\frac{1}{2}(\gamma\id+\vg\cdot\vsigma), \qquad 0 \leq g \leq \gamma \leq 2-g\, .
\end{equation}
With respect to a given parametrization, conditions \eqref{conditions on G++} can be recast into the following four inequalities:
\begin{eqnarray}
\no{\vg}&\leq& \gamma \, ,
\label{inequality 1}\\
\no{\va-\vg}&\leq&\alpha-\gamma \, ,
\label{inequality 2}\\
\no{\vb-\vg}&\leq&\beta-\gamma \, ,
\label{inequality 3}\\
\no{\va+\vb-\vg}&\leq&2+\gamma-\alpha-\beta \, .
\label{inequality 4}
\end{eqnarray}
In conclusion, effects $\A$ and $\B$ are coexistent if and only if there exist parameters $\gamma$ and $\vg$ such that the inequalities \eqref{inequality 1}--\eqref{inequality 4} are satisfied. [Note that the inequality in Eq.~\eqref{eq:definition of G++} is implied by these four inequalities, so we do not have to include Eq.~\eqref{eq:definition of G++} separately.] 

In the three dimensional space, each inequality can be viewed as a ball of allowed vectors $\vg$. These four balls are centered in points ${\bf 0}$, $\va$, $\vb$, and $\va+\vb$, respectively, with radii given by the right hand side of the corresponding inequality. The effects $A$ and $B$ are therefore coexistent if and only if there is a $\gamma$ such that the intersection of the four balls is non-empty. Important here is that the radii change with $\gamma$, which is a free parameter. The intersection also shows the freedom in choosing different vectors $\vg$ --- for each $\gamma$, when the intersection is non-empty, all points in the intersection can be chosen as $\vg$. From this also follows that a unique effect $\G_1$ satisfying Eq.~\eqref{eq:definition of G++} exists if and only if there is only one such $\gamma$ that the four balls intersect and for this particular $\gamma$, they intersect only in one point. 

By Proposition \ref{prop:unitary} in Appendix 1, the coexistence of $A$ and $B$ depends only on the numbers $\alpha,\beta,a,b$ and on the relative angle between $\va$ and $\vb$. Without any loss of generality, we choose the coordinate system such that the vector $\va$ lies along the $x$-axis and vector $\vb$ is in the $x$-$y$ plane. Then, a single point in the $x$-$y$ plane with coordinates $(x_0,y_0)$ represents a cone of three dimensional vectors $\vb$ (parametrizing effects $\B$), which all have the length along vector $\va$ equal to $\bx=x_0$ and the length in the perpendicular plane equal to $\by=y_0$.

Then, if there is a point $\vg$ in the intersection, its projection to the $x$-$y$ plane is also in the intersection, because the projection is closer than $\vg$ to each of the centers of the four balls. As we are interested in whether the intersection is empty or not, it is thus enough to study the intersection in the $x$-$y$ plane only. Projecting on the $x$-$y$ plane we obtain four circles centered in the corners of a parallelogram with sides $\va$ and $\vb$, which have the radii given in Eqs.~\eqref{inequality 1}--\eqref{inequality 4}. This geometrical formulation of the problem is summarized in Fig.~\ref{fig:jedna}.

\begin{figure}
\psfig{scale=0.5, file=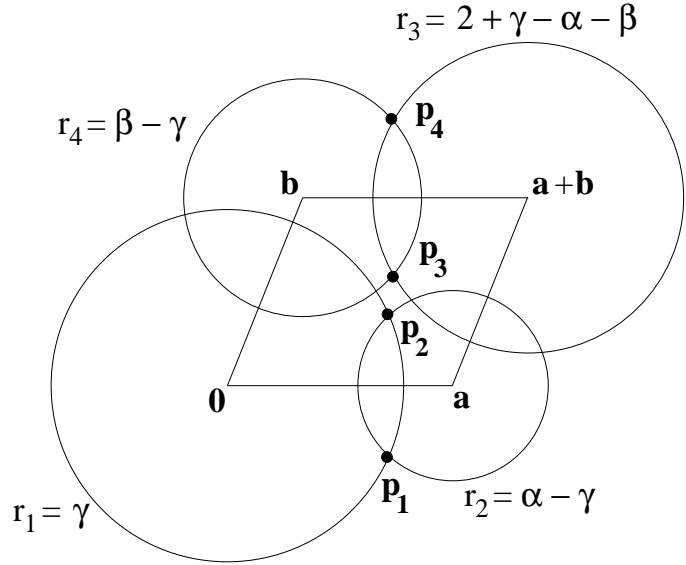}
\caption{The effects $A$ and $B$ are coexistent if and only if the intersection of the four circles is non-empty. The four circles are centered at the corners of a parallelogram with sides $\va$ and $\vb$. The circles' radii are given in the figure and depend on $\gamma$. For the particular $\gamma$ used in the figure, the intersection is empty. Later we will also need the common points of circles 1 and 2 denoted by ${\bf p_1}$ and ${\bf p_2}$, and circles 3 and 4, denoted by ${\bf p_3}$ and ${\bf p_4}$.
\label{fig:jedna}}
\end{figure}

\subsection*{Step 2: Restriction to the boundary cases}

We will answer the question of whether $\A$ and $\B$ are coexistent by fixing the parameters $\alpha$, $a$, and $\beta$ and specifying the allowed region $\allow$ in the two dimensional $x$-$y$ plane --- if a vector $\vb$ lies inside this allowed region, then the corresponding effects $\A$ and $\B$ are coexistent. 

It follows from Proposition \ref{prop:convex} in Appendix 1 that if vector $\vb$ is in the allowed region $\allow$, then all the vectors $\lambda\vb$, $0\leq \lambda\leq 1$, are in the allowed region $\allow$ also. Namely, choosing $B=\half (\beta\id+\vb\cdot\vsigma)$ and $C=\half \beta\id$ in Proposition \ref{prop:convex} we arrive at this conclusion. In later steps of this proof we will find a vector $\vb$ in each direction of the $x$-$y$ plane such that $\vb$ is in the allowed region $\allow$ but there is no vector $\vb'\in\allow$ having the same direction as $\vb$, but greater length. This set, which we call \emph{the boundary of the allowed region} $\allow$, thus characterizes all vectors in $\allow$.

The key property of the boundary which we exploit in our investigation is the following. 
\begin{proposition} The boundary of the allowed region $\allow$ is formed by such vectors $\vb$ that only such parameters $\gamma$ exist that the set of the intersection points of the four circles is non-empty but has zero area. 
\label{main idea}
\end{proposition}
\begin{proof} Assume that for a vector $\vb$ there is $\gamma$ such that the intersection has positive area. The boundaries of the circles move smoothly when changing vector $\vb$. Thus, there is $\epsilon>0$ such that for all vectors $\vb'$ satisfying $\|\vb'-\vb\|<\epsilon$, the change from $\vb$ to $\vb'$ does not make the intersection of the circles disappear. In particular, there is $\vb'$ having the same direction as $\vb$, but greater length. Therefore, $\vb$ is not in the boundary of $\allow$.
\end{proof}

Two circles can intersect in a set with a positive area or in a point. It follows that a non-empty zero area intersection of any number of circles is a point (and not, e.g. a curve).

Finally, we make an interesting observation (which is, however, not needed in the proof).
In the same way as previously,  we can deduce that if the intersection region has positive area, there must be an interval of $\gamma$'s leading to intersections. This means that a unique $\gamma$ (such that $\G$ exists) implies a single point intersection (and consequently unique $\vg$). It is, however, not true, that the existence of a single point intersection implies unique $\G$ --- for example, there are cases where there are only single point intersections, but $\gamma$ can be chosen from an interval of positive length (and also vectors $\vg$ differ for different $\gamma$). This fact will become evident later in the proof.

\subsection*{Step 3: Division of single point intersections into three cases}

Four circles can intersect in one point in three distinct ways:
\begin{itemize}
\item 2CI --- two circles intersect in one point and this point lies inside of the two remaining circles. We will see that one of the possible 2CI intersections defines the boundary of $\allow$, which is formed by vectors $\vb$ of length $\beta$, leading to conditions (C1) and (C2).
\item 3CI --- three circles intersect in one point, but out of these three, no two circles intersect in one point and the point lies strictly inside of the fourth circle. It will turn out that a 3CI case never defines a boundary point of $\allow$ --- a necessary condition for a 3CI implies that one of these three circles contains another one and therefore can be disregarded, forcing the case to be a 2CI case.
\item 4CI --- there is a point laying on the boundary of all four circles, but this point does not lead to 2CI. Such intersections define the boundary formed by vectors $\vb$ shorter than $\beta$, leading to condition (C3).
\end{itemize}
These three cases are illustrated in Fig.~\ref{tri}. In the following, we will put aside the case where $\va$ and $\vb$ are (anti)parallel vectors and $a=\alpha$. This assumption simplifies our investigation a bit and we will later in step 4 check that this case is also covered by the final result.

\begin{figure}
\psfig{width=\textwidth,file=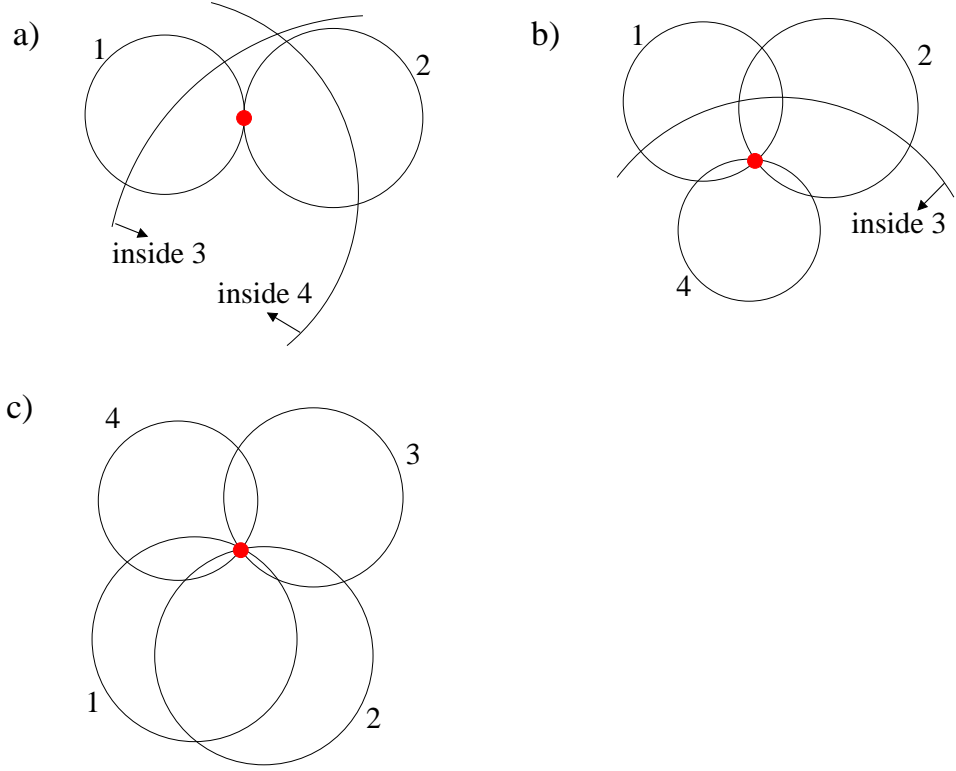}
\caption{Three different cases where the intersection is one single point. (a) Two circles touch in a single point, which is inside the remaining two circles. (b) Three circles intersect in a single point; it is not the case (a), and the intersection is inside the fourth circle. (c) Four circles intersect in a single point and it is not the case (a) or (b).
\label{tri}}
\end{figure}

\subsubsection*{Boundary of the allowed region in the 2CI case}

If $\vb$ is on the boundary of the allowed region and it is a 2CI case, then there is $\gamma_0$ such that two circles intersect in a point that is inside the remaining two circles and there is no $\gamma$ for which the intersection has a nonzero area. Then the two circles making the single point intersection cannot be the circles 1 and 3 as both these circles' radii grow with $\gamma$ --- we could enlarge $\gamma_0$ by a small amount such that the circles 1 and 3 would intersect in a nonzero area and the two remaining circles would not move enough to leave this area. We would thus obtain a four circle intersection of a nonzero area, which contradicts Proposition \ref{main idea}. Similarly, a boundary point formed by 2CI can not be due to the circles 2 and 4 --- they both grow if we decrease $\gamma$. Moreover, a boundary point can arise neither from circles 1 and 2, since they do not change with changing $\vb$, nor circles 3 and 4, which shift by the same amount when $\vb$ is changed. A single point intersection of circles 2 and 3 implies equality $b=2-\beta$, which is never the case since we have restricted ourselves to the case $\beta\leq 1$. The only possibility is that a single point intersection is formed by circles 1 and 4, leading to the condition $b=\beta$. 

We will now specify when a vector $\vb$ of length $\beta$ is in the allowed region $\allow$. In this case it is also on the boundary of $\allow$ since $b$ can not be larger than $\beta$ according to Eq.~\eqref{definition of B}. 
Assume that $b=\beta$. The radii of circles 1--4 are well defined (i.e. non-negative) when 
\begin{equation}\label{eq:gamma-basic-restriction}
0\leq\gamma\leq\min(\alpha,\beta)\, .
\end{equation}
For a given $\gamma$ satisfying Eq.~\eqref{eq:gamma-basic-restriction}, circles 1 and 4 touch in one point $\gamma\vb/b$. This point is inside of circle 2 if $\no{\gamma \vb/b-\va}\leq \alpha-\gamma$. We can write this requirement in the form
\begin{equation}
\gamma \leq \gamma_M:=\frac{\beta}{2}\frac{\alpha^2-a^2}{\alpha\beta-\va\cdot\vb} \, .
\end{equation} 
Here the number $\gamma_M$ is well defined and non-negative unless $\va$ and $\vb$ are parallel vectors and $a=\alpha$, a case which will be treated later. 

Similarly, one finds that the point $\gamma\vb/b$ is inside circle 3 when \begin{equation}
\gamma\geq\gamma_m := \frac{\beta}{2}\frac{(2-\alpha-\beta)^2-\no{\va+\vb}^2}{\alpha\beta-\va\cdot\vb-2\beta}\, .
\label{gamma min}
\end{equation} 
Again, the denominator in Eq.~\eqref{gamma min} is zero only in the case we have put aside. We conclude that the vector $\vb$ is in the allowed region if there exists $\gamma$ satisfying the inequalities \eqref{eq:gamma-basic-restriction} and
\begin{equation}\label{eq:gamma-m-M}
\gamma_m\leq\gamma\leq\gamma_M\, .
\end{equation}
We already noted that $\gamma_M \geq 0$. It also holds $\gamma_m\leq\min(\alpha,\beta)$. Namely, putting for $\va\cdot\vb$ the largest possible value $a\beta$, we get $\gamma_m=\frac{1}{2}(\alpha+ a)-(1-\beta)\leq\min(\alpha,\beta)$. Since $\gamma_m$ is an increasing function of $\va\cdot\vb$ (for all other parameters fixed), we have $\gamma_m\leq\min(\alpha,\beta)$ for all possible values of $\va\cdot\vb$. 
From this follows that $\gamma$ fulfilling both Eqs.~\eqref{eq:gamma-basic-restriction} and \eqref{eq:gamma-m-M} exists (and the vector $\vb$ is in the allowed region $\allow$) if and only if $\gamma_M-\gamma_m \geq 0$. 

We now look at the difference $\gamma_M-\gamma_m$ as a function of $\bx$. As we are only interested in the sign, we can equally well study the expression
\begin{equation}\label{eq:quadratic}
\left( \alpha - a\frac{\bx}{\beta} \right) \left( 2-\alpha + a\frac{\bx}{\beta} \right) \left( \gamma_M - \gamma_m \right)
\end{equation}
since the first two terms are positive unless $\bx=\pm\beta$ and $a=\alpha$. This expression is a quadratic polynomial of $\bx$, with roots 
\begin{equation}
\bx^\pm\equiv b_0\pm w:=\frac{1}{a}\left( (1-\alpha)(1-\beta) \pm \sqrt{D}\right)\, ,
\label{eq:limit bs}
\end{equation}
where
\begin{equation}
D:=(1-\alpha)^2-\beta[(1-\alpha)^2+1-a^2]+\beta^2\, .
\end{equation}
Here we obtained $b_0$ and $w$ used in Theorem \ref{th:fundamental}. The coefficient at $\bx^2$ of the quadratic polynomial \eqref{eq:quadratic} is positive and the polynomial is non-negative in the points $\bx=\pm\beta$. Therefore, there are three possible cases: (i) The discriminant $D$ is negative and the roots are complex --- this means that the difference $\gamma_M-\gamma_m$ is positive for all vectors $\vb$ of length $\beta$. (ii) Both roots are outside the interval $(-\beta,\beta)$. Again, this means that the difference $\gamma_M-\gamma_m$  is positive for all  $\vb$. (iii) Both roots are in the interval $[-\beta,\beta]$. In this case the difference $\gamma_M-\gamma_m$ is negative in between the two roots $\bx^\pm$, and these solutions do not correspond to a vector $\vb$ in $\allow$.

As the last step, we take a look at $D$ as a quadratic polynomial of $\beta$. First of all, it is non-negative at $\beta=0$ and $\beta=1$, and its two roots, labeled by $\beta_1$ and $\beta_2$  such that $\beta_1\leq\beta_2$, belong to the interval $[0,1]$. For $\beta\in[\beta_1,\beta_2]$, the discriminant $D$ is negative and therefore all $\vb$ are in the allowed region as discussed earlier. The scaling property from Proposition \ref{prop:scaling} says that all effects $B$ having $b=\beta\leq\beta_1$ will be coexistent with $A$ as well. Therefore, whenever $\beta\leq\beta_2$, all vectors $\vb$ of length $\beta$ are in the allowed region. What is left is to check the case $\beta>\beta_2$. 

We assume that $\beta>\beta_2$ and we show that in this case the solutions $\bx^\pm$ lie inside the interval $[-\beta,\beta]$. Let us look at the expression $a\bx^-$ as a function of $\alpha$, $a$, $\beta$.  We define a function $f$ by formula
\begin{equation*}
f(\alpha,a,\beta):= (1-\alpha)(1-\beta) - \sqrt{(1-\alpha)^2-\beta[(1-\alpha)^2+1-a^2]+\beta^2} \, ,
\end{equation*}
and the domain of $f$ is taken to be the region where $\alpha\in[0,1]$, $a\in[0,\alpha]$, and $\beta\in [\beta_2(a,\alpha),1]$. Then $f$ is a continuous function and its domain is a connected region in $\R^3$. A direct calculation shows that the equation $f(\alpha,a,\beta)=a\beta$ implies
\begin{equation}
0=\beta(1-\beta)(\alpha-a)(2-\alpha+a)\, .
\label{zero}
\end{equation}
Thus, $f$ can take the value $a\beta$ only on the boundary of its domain. From the continuity of $f$ and the connectedness of its domain then follows, that if for one point in the domain we have $f(\alpha,a,\beta) < a\beta$, then $f(\alpha,a,\beta) \leq a\beta$ in the whole domain. An analogous equation to Eq.~\eqref{zero} allows similar reasoning for the lower limit $-a\beta$. On the other hand, we have $f(2/3,1/2,4/5)\simeq -0.18$, which is inside the interval $(-a\beta,a\beta)=(-2/5,2/5)$. We thus conclude that $\bx^-\in[-\beta,\beta]$. The fact that $\bx^+\in[-\beta,\beta]$ can be shown in a similar way.

The fact that for $\beta\leq\beta_2$ all vectors $\vb$ are in the allowed region $\allow$ leads us to the definition of sharpness for effects in Eq.~(\ref{eq:sharpness}) -- we define the sharpness as $\S(A)=1-\beta_2$. We then conclude that 
\begin{itemize}
\item  if $\beta\leq1-\mathfrak{S}(\A)$, the whole boundary is formed by vectors $\vb$ of length $\beta$ -- in other words, in this case the allowed region is a circle with diameter $\beta$ and the center at $\vnull$, corresponding to (C1);
\item  if $\beta>1-\mathfrak{S}(\A)$, the boundary is given by vectors $\vb$ of length $\beta$ if and only if $\bx\notin (\bx^-,\bx^+)$, corresponding to (C2).
\end{itemize}

\subsubsection*{Boundary of the allowed region in the 3CI case}

Let us assume, for instance, that a 3CI case defining a boundary point $\vb$ is formed by the intersection of the circles 1, 2, and 4, i.e., the circles 1, 2 and 4 intersect in a single point which is inside the circle 3 (see Fig.~\ref{tri}). Looking at Fig.~\ref{fig:jedna}, points common to circles 1 and 2 are ${\bf p_1}$ and ${\bf p_2}$. The first one is not closer to circle 4 than the second one. Therefore, if the circles 1, 2, and 4 have a single common point, it must be ${\bf p_2}$.

Let us define the following function to compare the distance of $\bf p_2$ from the center of the circle 4 and its radius,
\begin{equation}
d(\gamma):=\no{\vb-{\bf p_2}(\gamma)}^2-(\beta-\gamma)^2 \, .
\label{eq:magic d}
\end{equation}
If the point ${\bf p_2}$ lies on circle 4 for some $\gamma_0$, then $d(\gamma_0)=0$. Moreover, if the point ${\bf p_2}$ lies inside (outside) circle 4, then  $d(\gamma_0)<0$ [$d(\gamma_0)>0$].  If $\partial_\gamma d(\gamma)|_{\gamma_0} < 0$, then there exists an interval ($\gamma_0$,$\gamma_1$) where $d(\gamma)<0$. Since we have assumed that the common point of circles 1, 2, and 4 is inside circle 3, there exists $\gamma\in(\gamma_0,\gamma_1)$ such that the four circles intersect in a region with nonzero area. This is in contradiction with Proposition \ref{main idea}. A similar reasoning rules out the case $\partial_\gamma d(\gamma)|_{\gamma_0} > 0$. Therefore, a necessary condition for a 3CI case is the set of equations $d(\gamma_0)=0$ and $\partial_\gamma d(\gamma)|_{\gamma_0}=0$. 

The coordinates for ${\bf p_2}$ are $x=\alpha(2\gamma-\alpha)/2a+a/2$ and $y=\sqrt{\gamma^2-x^2}$. By making the substitution $\gamma= (\alpha-a \kappa)/2$ we can express the distance $d$ as a function of the new variable $\kappa$ in the form
\begin{equation}
d(\kappa)=c_1(c_2-\sqrt{1-\kappa^2}+c_3 \kappa)\, ,
\end{equation}
where $c_1=\by\sqrt{\alpha^2-a^2}$, $c_1 c_2=b^2-\va\cdot\vb+\alpha\beta-\beta^2$, and $c_1c_3=\bx\alpha-a\beta$. Equation $\partial_\gamma d(\gamma)|_{\gamma_0}=0$ leads to a unique solution $\kappa_0=-c_3/\sqrt{1+c_3^2}$. Putting this into equation $d(\gamma_0)=0$, we get $c_2^2-c_3^2-1=0$. Substituting back the definitions for $c_1$, $c_2$ and $c_3$ we finally obtain that a necessary condition for this particular 3CI is
\begin{equation}
(b^2-\beta^2)[\no{\va-\vb}^2-(\alpha-\beta)^2]=0\, .
\label{3CI condition}
\end{equation}
If the expression in the first bracket is zero, we obtain the condition for 2CI of circles 1 and 4. If the second bracket is zero, circles 2 and 4 are one inside the other (if $\alpha \geq \beta$, then circle 4 is inside circle 2, and it is the opposite if $\alpha\leq \beta$). Their intersection is then the whole smaller circle and this fact does not depend on $\gamma$. We can then disregard the larger circle completely, because the intersection does not depend on it in any respect. The 3CI is thus reduced to 2CI and can not therefore define boundary points different from those found in the previous section dealing with 2CI. In the same way, one finds out that the three other possible 3CI cases are similar and always lead to boundary points defined by a 2CI intersection. The resulting conditions, analogous to Eq.~\eqref{3CI condition}, are summarized in Table~\ref{tab1} for all four possible 3CI.

\begin{table}
\caption{\label{tab1}Necessary conditions for all four possible three circle intersections. The three circles intersecting in a single point are given in the first column. The necessary condition and its geometrical meaning for a 3CI defining the boundary are in the second column.}
\begin{tabular}{|l|l|}
\hline
3CI&necessary condition\\
&and its geometrical meaning\\
\hline
1, 2, and 3&$[b^2-(2-\beta)^2][\no{\va+\vb}^2-(2-\alpha-\beta)^2]=0$\\
& 1 and 3 are one inside the other\\
\hline
1, 2, and 4&$[b^2-\beta^2][\no{\va-\vb}^2-(\alpha-\beta)^2]=0$\\
& 1 and 4 are touching, or 2 and 4 are one inside the other\\
\hline
1, 3, and 4&$[b^2-\beta^2][\no{\va+\vb}^2-(2-\alpha-\beta)^2]=0$\\
& 1 and 4 are touching, or 1 and 3 are one inside the other\\
\hline
2, 3, and 4&$[b^2-(2-\beta)^2][\no{\va-\vb}^2-(\alpha-\beta)^2]=0$\\
& 2 and 4 are one inside the other\\
\hline
\end{tabular}
\end{table}

\subsubsection*{Boundary of the allowed region in the 4CI case}

Let us assume that the first two conditions in (C3) hold. We show that then the right-hand side of Eq.~\eqref{nontrivial condition} defines the perpendicular component of vectors $\vb$ forming the boundary of the allowed region $\allow$. 

A four point intersection can occur if one of the points ${\bf p_1}$ and ${\bf p_2}$ coincides with one of the points ${\bf p_3}$ and ${\bf p_4}$. Since $\by\geq 0$, a single point intersection must be such that points ${\bf p_2}$ and ${\bf p_3}$ coincide. Putting their $x$ coordinates to be equal we obtain the solution for $\gamma$,
\begin{equation}
\gamma=\frac{1}{2}\left[\va\cdot\vb+\alpha\beta-2(1-\alpha)(1-\beta) \right]\, .
\label{gamma ac}
\end{equation}
This solution represents the four circle intersection if the intersection points ${\bf p_2}$ and ${\bf p_3}$ exist. Point ${\bf p_2}$ exists if $(\alpha-a)/2\leq\gamma\leq(\alpha+a)/2$, while point ${\bf p_3}$ exists if $(\alpha-a)/2-(1-\beta)\leq\gamma\leq(\alpha+a)/2-(1-\beta)$. Using these conditions, we conclude that Eq.~\eqref{gamma ac} represents a case ${\bf p_2}$=${\bf p_3}$ if and only if the following condition is fulfilled:
\begin{equation}\label{eq:p2=p3}
(1-\alpha)(1-\beta) - (a+\beta-1) \leq \va\cdot\vb \leq (1-\alpha)(1-\beta) + (a+\beta-1)\, .
\end{equation}

Under the first two conditions in (C3), these inequalities are fulfilled. First of all, a straightforward calculation shows that $\S(\alpha,a)\leq a$. Hence, the inequality $\beta > 1- \S(\alpha,a)$ guarantees that $a+\beta-1>0$. It is then easy to verify that the $\sqrt{D}\leq a+\beta -1$. Therefore, if $\mod{\bx-b_0}<w$, then Eq.~\eqref{eq:p2=p3} holds.

Putting equal the $y$ coordinates for the points ${\bf p_2}$ and ${\bf p_3}$ we finally obtain the equation of the coordinate $\by$ as a function of $\bx$,
\begin{equation}
\begin{split}
\by=&\frac{1}{2a}\sqrt{(\alpha^2-a^2)\{a^2-[(2-\alpha)(1-\beta)+a \bx]^2\}}\\+
&\frac{1}{2a}\sqrt{((2-\alpha)^2-a^2)\{a^2-[\alpha(1-\beta)+a \bx]^2\}} \, .
\label{blue curve}
\end{split}
\end{equation}
This can be rewritten in the form given in Eq.~\eqref{nontrivial condition}. 

\subsection*{Step 4: the case of parallel vectors $\va$ and $\vb$}

Finally, we look at the situation where the vectors $\va$ and $\vb$ are parallel. In this case the effects $A$ and $B$ commute, and this implies that they are coexistent. To check their coexistence directly from Definition \ref{def:definition of coexistence}, one can use a four outcome observable $\G$ defined as
\begin{equation*}
\G_1=AB,\ \G_2=A(\id-B),\ \G_3=(\id-A)B,\ \G_4=(\id-A)(\id-B)\, .
\end{equation*}

On the other hand, the fact that $\va$ and $\vb$ are parallel means that $\by=0$. Clearly, the conditions (C1)--(C3) do not then lead to any restrictions. This completes the proof of Theorem \ref{th:fundamental}.

\subsection*{Step 5: Proof of inequality \eqref{inside circle}}

We have seen in step 3 that for $|\bx-b_0|< w$, condition \eqref{eq:p2=p3} is fulfilled and therefore the expression in Eq. \eqref{blue curve} determines a vector $\vb$ in the allowed region $\allow$. We know from the 2CI case that for $|\bx-b_0|< w$, vector $\vb$ of length $\beta$ is not in the allowed region. From this follows that the length of vector $(\bx, \by^{max})$ must be shorter than $\beta$. We can also see this directly from the expression for $\by^{max}$. If we denote $r^2=\bx^2+(\by^{max})^2$, we get after some algebraic manipulation that
\begin{itemize}
\item $r=\beta$ implies $|\bx-b_0|= w$,
\item $\partial_{\bx} r=0$ implies that either $|\bx-b_0|= w$ or $|\bx-b_0|= 0$. 
\end{itemize}
The first point shows that the vector $(\bx,\by^{max})$ does not reach the length of $\beta$ anywhere inside the interval $\bx\in(\bx^-,\bx^+)$ and therefore the inequality in Eq.~\eqref{inside circle} holds. On the other hand, since $r=\beta$ in the two points $\bx=\bx^\pm$, the continuity of $r$ as a function of $\bx$ implies that it reaches the minimum value somewhere inside the interval. The second point shows that it happens at $\bx=b_0$.

\section*{Acknowledgements}

Our work has been supported by projects CONQUEST No.~MRTN-CT-2003-505089, QAP No.~2004-IST-FETPI-15848, and APVV No.~RPEU-0014-06.

%\bibliographystyle{unsrt}
%\bibliography{bibliography}

\begin{thebibliography}{10}

\bibitem{Chefles00}
A.~Chefles.
\newblock Quantum state discrimination.
\newblock {\em Contemporary Physics}, 41:401--424, 2000.

\bibitem{BuCaLa95}
P.~Busch, G.~Cassinelli, and P.J. Lahti.
\newblock Probability structures for quantum state spaces.
\newblock {\em Rev. Math. Phys.}, 7:1105--1121, 1995.

\bibitem{FQMI83}
G.~Ludwig.
\newblock {\em Foundations of Quantum Mechanics I}.
\newblock Springer-Verlag, New York, 1983.

\bibitem{SEO83}
K.~Kraus.
\newblock {\em States, Effects, and Operations}.
\newblock Springer-Verlag, Berlin, 1983.

\bibitem{BeBr84}
C.H. Bennet and G.~Brassard.
\newblock In {\em Proceedings of IEEE International Conference on Computers,
  Systems and Signal Processing}, pages 175--179, New York, 1984. IEEE.

\bibitem{Werner01}
R.F. Werner.
\newblock Quantum information theory -- an invitation.
\newblock In {\em Quantum Information: an Introduction to Basic Theoretical
  Concepts and Experiments}, chapter~2, pages 14--57. Springer-Verlag, 2001.

\bibitem{AnBaAs05}
E.~Andersson, S.M. Barnett, and A.~Aspect.
\newblock Joint measurements of spin, operational locality, and uncertainty.
\newblock {\em Phys. Rev. A}, 72:042104, 2005.

\bibitem{BuSh06}
P.~Busch and C.~Shilladay.
\newblock Complementarity and uncertainty in mach--zehnder interferometry and
  beyond.
\newblock {\em Phys. Rep.}, 435:1--31, 2006.

\bibitem{BuHe07}
P.~Busch and T.~Heinosaari.
\newblock Approximate joint measurements of qubit observables.
\newblock arXiv:0706.1415v2 [quant-ph], 2007.

\bibitem{LiLiYuCh07}
Nai-Le Liu, Li~Li, Sixia Yu, and Zeng-Bing Chen.
\newblock Complementarity enforced by joint measurability of unsharp
  observables.
\newblock arXiv:0712.3653v1 [quant-ph].

\bibitem{QTOS76}
E.B. Davies.
\newblock {\em Quantum Theory of Open Systems}.
\newblock Academic Press, London, 1976.

\bibitem{PSAQT82}
A.S. Holevo.
\newblock {\em Probabilistic and Statistical Aspects of Quantum Theory}.
\newblock North-Holland Publishing Co., Amsterdam, 1982.

\bibitem{OQP97}
P.~Busch, M.~Grabowski, and P.J. Lahti.
\newblock {\em Operational Quantum Physics}.
\newblock Springer-Verlag, Berlin, 1997.
\newblock second corrected printing.

\bibitem{FQMEA02}
W.M. de~Muynck.
\newblock {\em Foundations of Quantum Mechanics, an Empiricist Approach}.
\newblock Kluwer Academic Publishers, Dordrecht, 2002.

\bibitem{Busch86}
P.~Busch.
\newblock Unsharp reality and joint measurements for spin observables.
\newblock {\em Phys. Rev. D}, 33:2253--2261, 1986.

\bibitem{Molnar01b}
L.~Moln{\'a}r.
\newblock Characterizations of the automorphisms of {H}ilbert space effect
  algebras.
\newblock {\em Comm. Math. Phys.}, 223:437--450, 2001.

\bibitem{BuSc08}
P.~Busch and H.-J. Schmidt.
\newblock Coexistence of qubit effects.
\newblock arXiv:0802.4167v2 [quant-ph], 2008.

\bibitem{YuLiLiOh08}
S.~Yu, N.~Liu, L.~Li, and C.H.~Oh.
\newblock e-print arXiv:0805.1538v1.

\end{thebibliography}

\end{document}